\documentclass{article}
\usepackage[utf8]{inputenc} 
\usepackage[T1]{fontenc}    
\usepackage{hyperref}       
\usepackage{url}            
\usepackage{booktabs}       
\usepackage{amsfonts}       
\usepackage{nicefrac}       
\usepackage{microtype}
\usepackage{authblk}
\usepackage[margin=1in]{geometry}
\usepackage{amsmath}
\usepackage{amssymb}
\usepackage{amsfonts}
\usepackage{amsthm}
\usepackage{mathrsfs}
\usepackage{bbm}
\usepackage{xcolor}
\usepackage[round,authoryear]{natbib}
\usepackage{graphicx}
\usepackage{subcaption}
\usepackage{algorithm}
\usepackage{algpseudocode}
\usepackage{makecell}
\hypersetup{
    colorlinks=true,
    linkcolor=blue,
    filecolor=magenta,      
    urlcolor=blue,
    citecolor=blue,
}
\newtheorem{thm}{Theorem}
\def\Real{\mathop{\mathbb{R}}\nolimits}

\def\argmin{\mathop{\rm argmin}\nolimits}

\newcommand{\ba}{\boldsymbol{a}}
\newcommand{\bb}{\boldsymbol{b}}
\newcommand{\bc}{\boldsymbol{c}}

\newcommand{\br}{\boldsymbol{r}}

\newcommand{\bu}{\boldsymbol{u}}

\newcommand{\bw}{\boldsymbol{w}}
\newcommand{\bx}{\boldsymbol{x}}
\newcommand{\by}{\boldsymbol{y}}
\newcommand{\bz}{\boldsymbol{z}}
\newcommand{\bmu}{\boldsymbol{\mu}}

\newcommand{\bM}{\boldsymbol{M}}

\newcommand{\bU}{\boldsymbol{U}}

\newcommand{\bW}{\boldsymbol{W}}
\newcommand{\bX}{\boldsymbol{X}}

\newcommand{\bbeta}{\beta}

\newcommand{\A}{\mathcal{A}}
\newcommand{\sP}{\mathbbm{P}}
\newcommand{\Pn}{\mathbbm{P}_n}
\newcommand{\one}{\mathbbm{1}}
\title{Principal Ellipsoid Analysis (PEA): Efficient non-linear dimension reduction \& clustering}
\author[1]{Debolina Paul}
\author[2]{Saptarshi Chakraborty}
\author[3]{Didong Li}
\author[3,4]{David Dunson\thanks{correspondence to: dunson@duke.edu}}
\affil[1]{Indian Statistical Institute, Kolkata}
\affil[2]{Department of Statistics, University of California, Berkeley}
\affil[3]{Department of Mathematics, Duke University}
 \affil[4]{Department of Statistical Science, Duke University}
\date{ }
\begin{document}
\maketitle
\begin{center}

\normalsize
\end{center}

\begin{abstract}
Even with the rise in popularity of over-parameterized models, simple dimensionality reduction and clustering methods, such as PCA and k-means, are still routinely used in an amazing variety of settings.  A primary reason is the combination of simplicity, interpretability and computational efficiency.  The focus of this article is on improving upon PCA and k-means, by allowing non-linear relations in the data and more flexible cluster shapes, without sacrificing the key advantages.  The key contribution is a new framework for Principal Elliptical Analysis (PEA), defining a simple and computationally efficient alternative to PCA that fits the best elliptical approximation through the data.  We provide theoretical guarantees on the proposed PEA algorithm using Vapnik-Chervonenkis (VC) theory to show strong consistency and uniform concentration bounds.  Toy experiments illustrate the performance of PEA, and the ability to adapt to non-linear structure and complex cluster shapes.  In a rich variety of real data clustering applications, PEA is shown to do as well as k-means for simple datasets, while dramatically improving performance in more complex settings.
\end{abstract}
{\bf Keywords:} Nonlinear Dimensionality Reduction, Manifold Learning, Clustering, Statistical Learning Theory.
\section{Introduction}
Clustering of data into groups of relatively similar observations is one of the canonical tasks in unsupervised learning.  With an increasing focus in recent years on very richly parameterized models,
there has been a corresponding emphasis in the literature on complex clustering algorithms.  A popular theme has been on clustering on the latent variable level, while allowing estimation of both the clustering structure and a complex nonlinear mapping from the latent to observed data level.  Such methods are appealing in
being able to realistically generate data that are indistinguishable from the observed data, while clustering observations in a lower-dimensional space.  

A particularly popular strategy is to develop clustering algorithms based on variational autoencoders (VAEs).  For example, instead of drawing the latent variables in a VAE from standard Gaussian distributions, one can use a mixture of Gaussians for model-based clustering \citep{dilokthanakul2016deep,lim2020deep,yang2019deep}.
The problem with this family of methods is that, with a rich enough deep neural network, VAEs can accurately approximate any data generating distribution regardless of the continuous density placed on the latent variables.  If one uses a richer family of densities, such as a mixture model, then one can potentially approximate the data distribution using a simpler neural network structure.  However, the inferred clusters are not reliable due to problems of non-identifiability.  There are no guarantees that observations in the same cluster on the latent variable level are also close on the observed data level.  The clustering solutions are brittle and lack interpretability.  

A major focus of clustering is on obtaining interpretable groups in the data.  For this reason, in sharp contrast to the increasing focus in predictive contexts on highly complex and heavily parameterized models, in conducting clustering most practitioners continue to use very simple algorithms like $k$-means.  The goal of this paper is to define a new class of simple and computationally efficient nonlinear modifications of Principal Components Analysis (PCA), and to use this new class to obtain a competitor to $k$-means for clustering.  We refer to our class of algorithms as Principal Elliptical Analysis (PEA). In contrast to the focus of PCA on linear structure, we use ellipses as a simple geometric object used to approximate low-dimensional structure within each cluster.

A key problem with $k$-means is that the cluster shapes are intrinsically assumed to be spherical; this can lead to breaking up true clusters that have a more complex shape into spherical subclusters.  Attempting to allow more complex within-cluster data generating distributions has motivated a huge literature on highly complex approaches.  However, the within-cluster data distribution should not be multimodal or arbitrarily complex, or there is no need to have separate clusters in the data and non-identifiability and/or brittleness problems arise.  Hence, it is a mathematical necessity to restrict cluster shapes in order to allow reliable estimation of clusters.  Our conjecture is that data within a cluster can be accurately characterized as error-prone observations around the surface of an unknown ellipse.  This motivates elliptical modifications of PCA and $k$-means.

There is a rich existing literature developing methods to improve upon $k$-means, addressing well known problems in terms of 
sensitivity to high-number of clusters, inefficacy in high-dimensions and poor performance in the presence of linearly non-separable clusters.  Variants of $k$-means designed to address these problems include (among others) Weighted $k$-means \citep{kerdprasop2005weighted,huang2005automated}, Spherical $k$-means \citep{dhillon2001concept}, Fuzzy $C$-means \citep{dunn1974graph}, $k$-means++ \citep{arthur2007proceedings}, Bisecting $k$-means \citep{steinbach2000comparison}, X-means \citep{pelleg2000x}, G-means \citep{hamerly2004learning}, Power $k$-means \citep{xu2019power,chakraborty2020entropy}, Sparse $k$-means \citep{witten2010framework} and optimal trasport-based methods \citep{lacombe2018large,chakraborty2020hierarchical}. Recent development in this area include sparse subspace clustering methods \citep{SSC,8447505,yang2016ell}, which are known to be efficient for datasets that lie on a union of subspaces. This paradigm uses the fact that each data point can be written as an affine combination of all other points and thus, searching for the sparsest combination, automatically gives rise to  a representation of points lying in the same subspace. However, these variations either fail to capture non-linear structure of the data or are computationally complex, lacking interpretability and reproducibility.


There is also a rich literature generalizing PCA to allow non-linear structure including Curvilinear PCA \citep{demartines1997curvilinear}, Spherical PCA \citep{li2017efficient}, Principal Curves \citep{hastie1989principal}, Principal Manifolds \citep{gorban2007principal}, Generalized PCA \citep{vidal2005generalized}, Kernel PCA \citep{scholkopf1998nonlinear}, etc. Another class of methods, which includes the popular $t$-SNE approach \citep{maaten2008visualizing}, ISOMAP \citep{tenenbaum2000global}, Locally Linear Embedding (LLE) \citep{roweis2000nonlinear}, Laplacian Eigenmaps \citep{belkin2002laplacian}, Diffusion Maps \citep{coifman2006diffusion} and Autoencoders \citep{bengio2006nonlocal}, tries to reconstruct the data, while preserving some of its properties, such as distances between two points or the local linear structure of the dataset. Although excellent for data visualization, these are often computationally expensive and the new features often carry no interpretation in the original feature space.

Our Principal Elliptical Analysis (PEA) class of methods improves upon the above literature through a conceptually and computationally simple class of geometrically-based optimization algorithms. Data within each cluster are assumed to fall close to a local region of the surface of an ellipsoid.
The objective function is minimized via a simple block-coordinate descent algorithm. We provide uniform concentration bounds and assert the strong consistency of the obtained empirical solutions under some standard regularity conditions. The theoretical analysis is conducted appealing to Vapnik-Chervonenkis (VC) theory \citep{devroye2013probabilistic}, which although extensively used in  classification contexts, remains largely unexplored in finding empirical risks in unsupervised learning problems.

The rest of the paper is organized as follows. After discussing a motivating example in Section \ref{motiv}, we formulate Principal Ellipsoid Analysis (PEA) in Section \ref{problem formulation}. We discuss the block-coordinate descent algorithm to minimize the proposed objective function in Section \ref{optimization}. Section \ref{theo} provides uniform concentration bounds and strong consistency properties of the global optimal solutions. Detailed experimental analysis on a suite of real datasets is provided in Section \ref{expo}, followed by a brief discussion in Section \ref{con}. Codes for our algorithms are available at \url{https://github.com/DebolinaPaul/PEA}.


\section{Motivating Example}\label{motiv}
This section provides a motivating example demonstrating the efficacy of the proposed method. We simulated a dataset, which contains $200$ points, $100$ points in each cluster, in $2$ dimensions (for ease in visualization). Each of the clusters is generated from part of the ellipse, $\frac{x^2}{4^2}+\frac{y^2}{3^2}=1$. To generate the first cluster from the ellipse, we consider its parametric equation,
\begin{equation}\label{data_gen}
    x=4\cos\theta, \quad y=3\sin\theta.
\end{equation}
We sample $\theta \sim Unif(\frac{\pi}{4},\frac{3\pi}{4})$ and obtain the data points using equation \eqref{data_gen} and adding a $\mathcal{N}(0,0.3)$ noise to each coordinate. The second cluster is simulated in a similar manner but with the second coordinate shifted by $2$ units in the positive $y$ direction. The entire dataset is then centered and scaled (i.e. $z$-transformed). We run classical $k$-means \citep{macqueen1967some}, Weighted $k$-means \citep{huang2005automated}, sparse $k$-means \citep{witten2010framework}, and IF-HCT-PCA \citep{jin2016influential} on the dataset. For fair comparison, we also run $k$-means after influential feature selection \citep{jin2016influential}, which we call IF-HCT $k$-means. The scatterplots of the dataset color-coded with partitions obtained from different algorithms are illustrated in Figure \ref{m}. It is easily observed that among all its competitors, only the proposed PEA Clustering successfully partitions the two clusters.
\begin{figure*}[ht]
    \centering
    \begin{subfigure}[t]{0.3\textwidth}
    \centering
        \includegraphics[height=0.8\textwidth,width=0.9\textwidth]{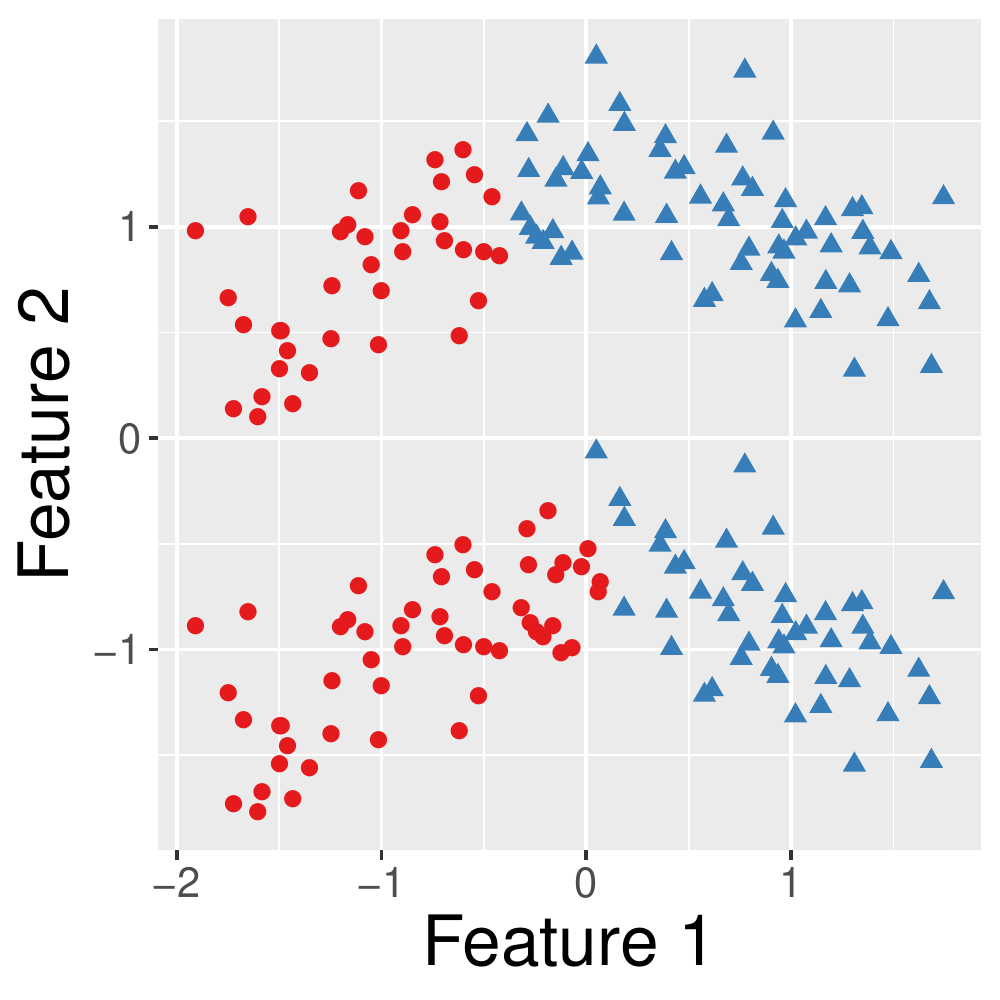}
        \caption{$k$-Means}
        \label{m1}
    \end{subfigure}
    ~
    \begin{subfigure}[t]{0.3\textwidth}
    \centering
        \includegraphics[height=0.8\textwidth,width=0.9\textwidth]{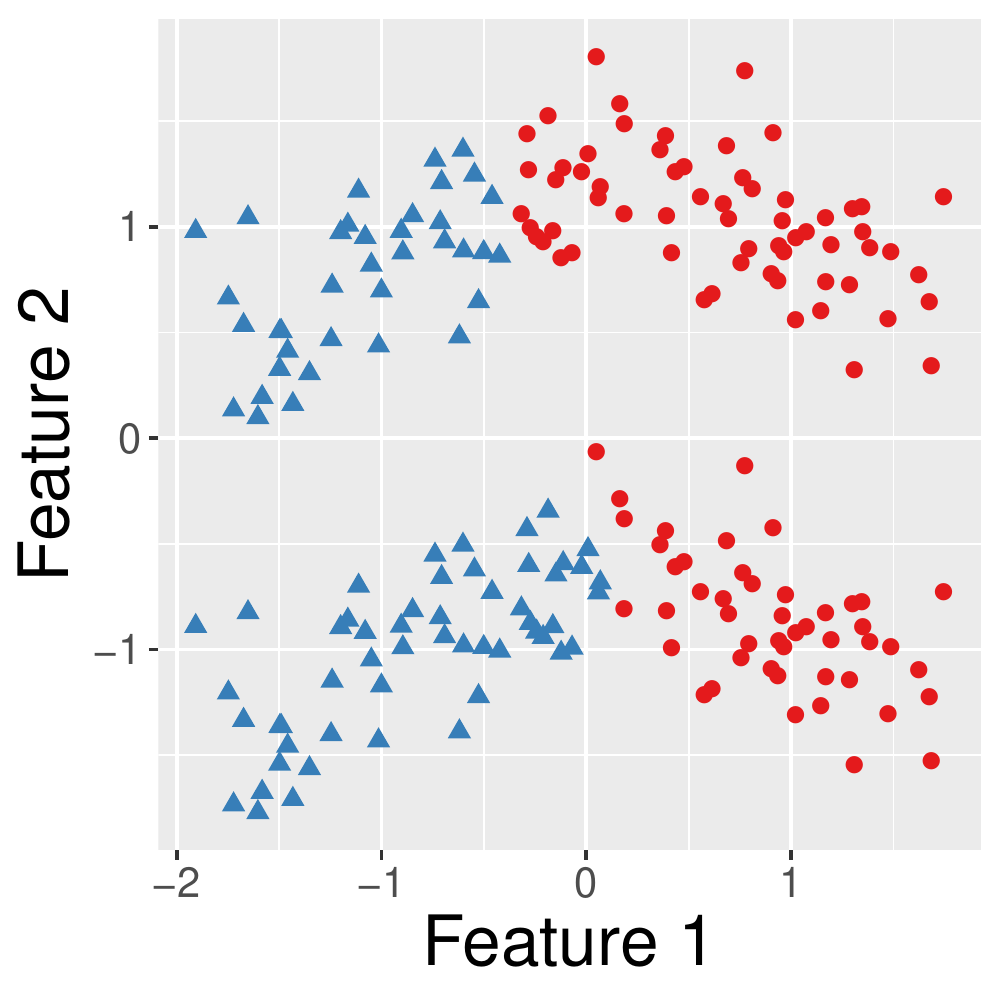}
        \caption{IF-HCT $k$-Means}
        \label{m2}
    \end{subfigure}
    ~
    \begin{subfigure}[t]{0.3\textwidth}
    \centering
        \includegraphics[height=0.8\textwidth,width=0.9\textwidth]{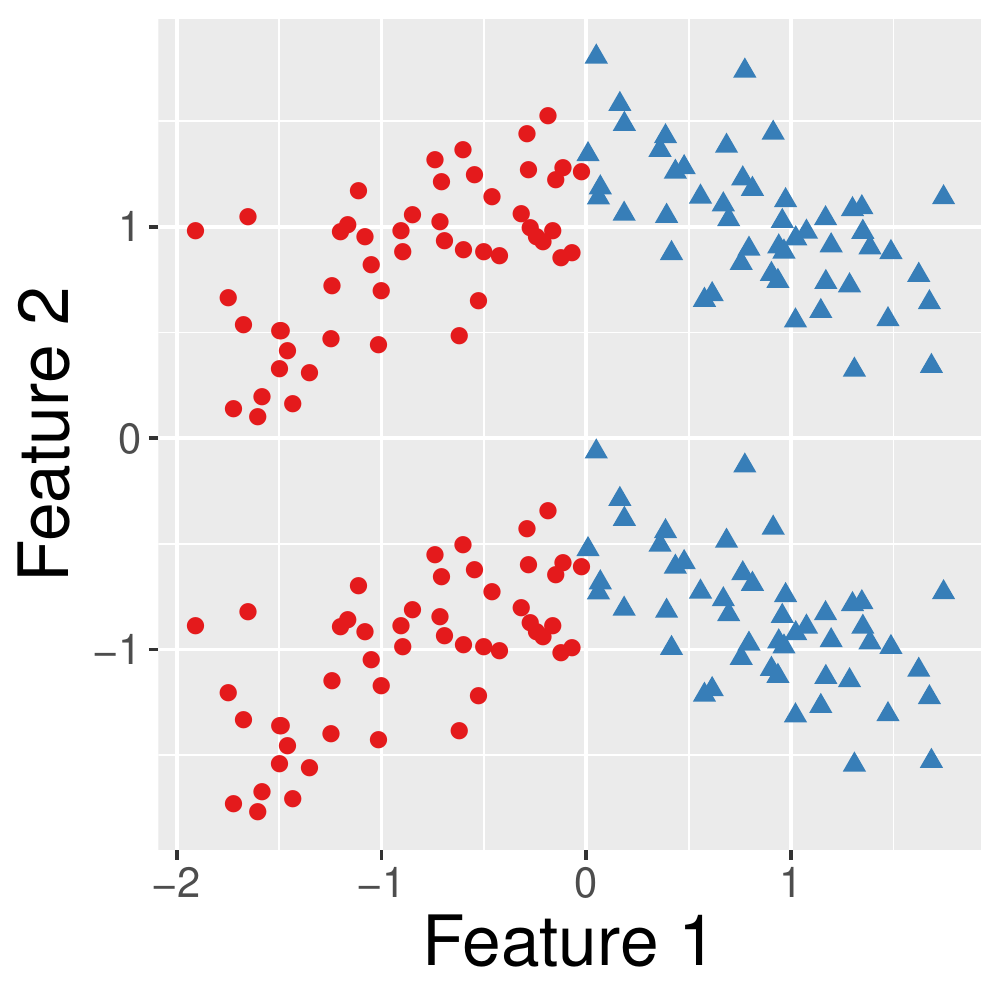}
        \caption{$Wk$-Means}
        \label{m3}
    \end{subfigure}
    ~
    \begin{subfigure}[t]{0.3\textwidth}
    \centering
        \includegraphics[height=0.8\textwidth,width=0.9\textwidth]{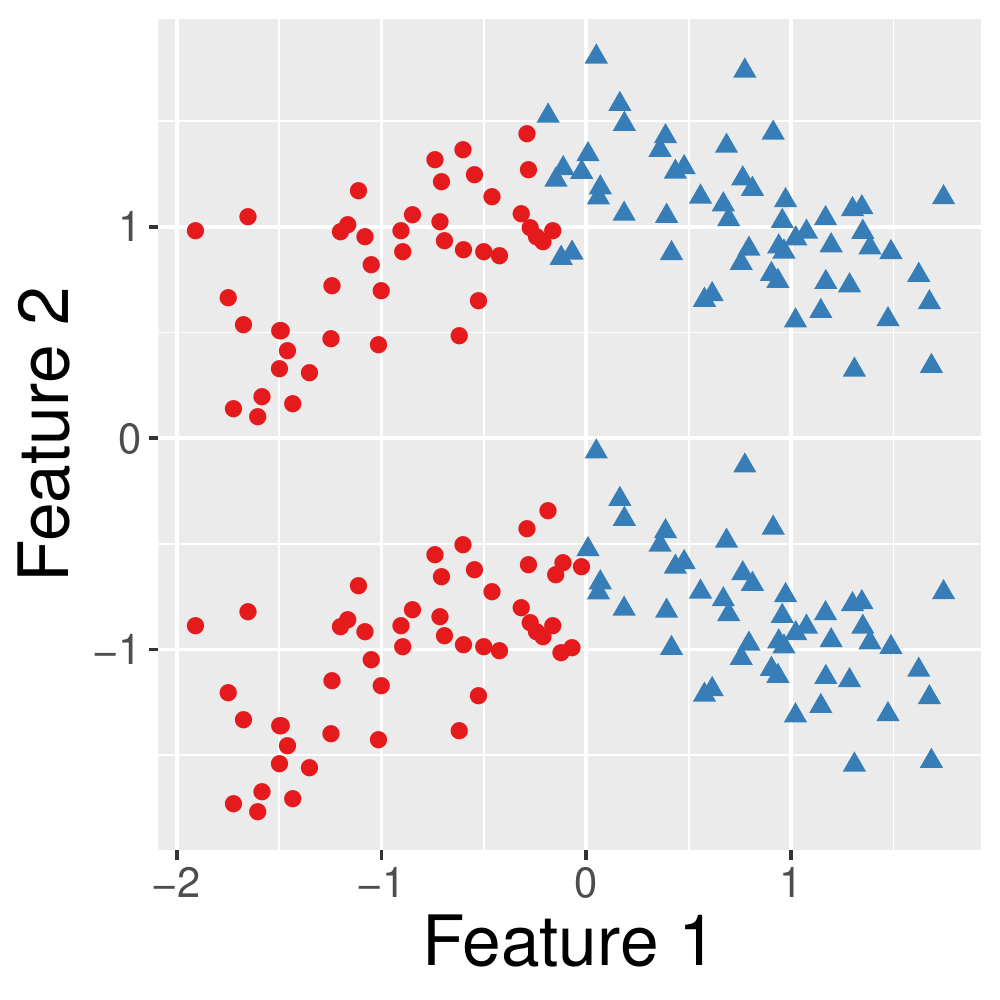}
        \caption{Sparse $k$-Means}
        \label{m4}
    \end{subfigure}
    ~
    \begin{subfigure}[t]{0.3\textwidth}
    \centering
        \includegraphics[height=0.8\textwidth,width=0.9\textwidth]{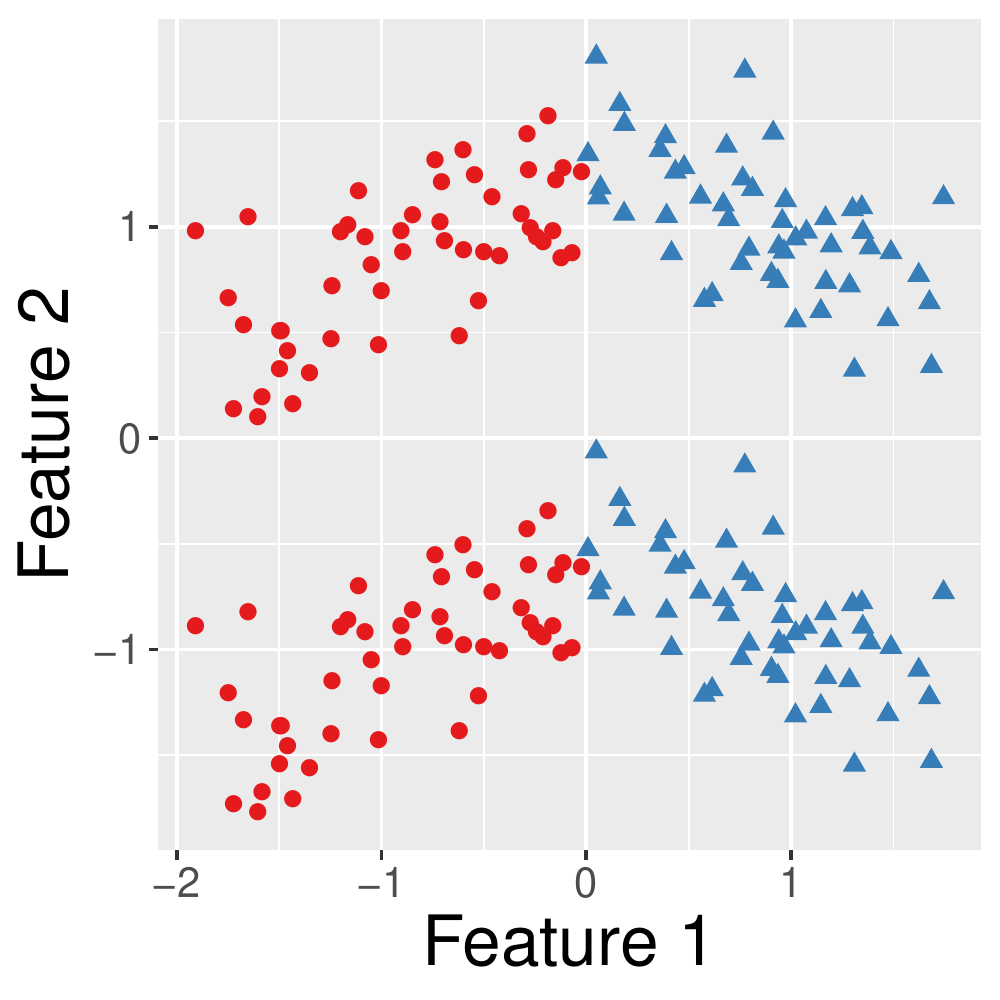}
        \caption{IF-HCT PCA}
        \label{m5}
    \end{subfigure}
    ~
    \begin{subfigure}[t]{0.3\textwidth}
    \centering
        \includegraphics[height=0.8\textwidth,width=0.9\textwidth]{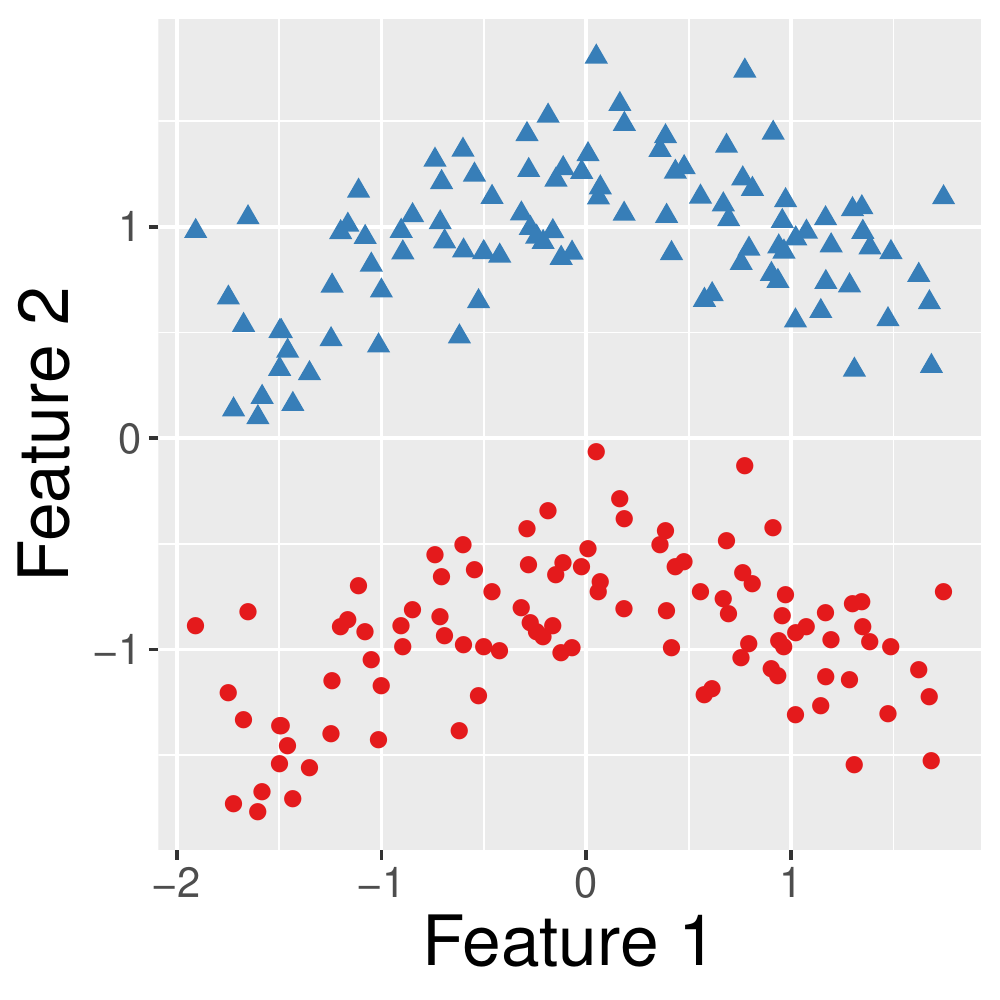}
        \caption{PEA Clustering}
        \label{m6}
    \end{subfigure}
    \caption{ Performance of different peer algorithms on the motivating example, showing the efficacy of PEA clustering. }
    \label{m}
\end{figure*}

\section{Principal Ellipsoid Analysis}
In this section, we formulate the proposed objective function for Principal Ellipsoid Analysis (PEA) and derive a coordinate-descent algorithm to minimize the PEA objective. 
\subsection{Problem Formulation}\label{problem formulation}
Let $\mathcal{X}=\{\bx_1,\dots,\bx_n\}$ be a data cloud of size $n$ in $\Real^p$. We assume that the points $\bx_1,\dots,\bx_n$ lie in the proximity of the surface of an axis parallel ellipse. Let the center of this ellipse be $\bmu$ and the length of the $p$ axes of the ellipse be $r_1,\dots,r_p$, respectively. In order to avoid trivial ellipsoids, we will assume that $0< \frac{1}{\Lambda} \le r_l \le \frac{1}{\lambda}$, for all $l =1,\dots,p$. One can expect that $\bx_i$ is generated by adding a random noise to the vector $\bmu + \br \circ \bu_i$, where $\|\bu_i\|_2=1$, $\br=(r_1,\dots,r_p)^\top$ and $\ba\circ \bb$ denotes the Hadamard product between $\ba$ and $\bb$. If we take $w_l=\frac{1}{r_l}$, $l=1,\dots,p$ and write $\bw=(w_1,\dots,w_p)^\top$, then $\bw \circ (\bx_i-\bmu) \approx \bu_i$. Here $\lambda \le w_l \le \Lambda$, for all $l=1,\dots,p$. We propose to estimate $\bw,\bmu$ and $\bu_i$ by minimizing the sum of squares of this error, which is given by, 
\begin{equation}\label{eqq}
    f(\bmu, \bw, \bU) = \sum_{i=1}^n \|\bw \circ (\bx_i-\bmu) - \bu_i\|_2^2,
\end{equation}
where $\bU_{n \times p}=[\bu_1:\dots:\bu_n ]^\top$. The objective function \eqref{eqq} is to be minimized subject to the constraint that $\bw \in \mathcal{C}$, where $\mathcal{C}=[\lambda,\Lambda]^p$ and $\bu_i^\top \bu_i=1$, for all $i=1,\dots,n$. We thus propose Principal Elliptical Analysis (PEA) as the minimization of the following optimization problem:
\begin{equation}
\label{e2}
\min_{\bmu, \bw, \bU}\sum_{i=1}^n \|\bw \circ (\bx_i-\bmu) - \bu_i\|_2^2
\end{equation}
subject to $\bw \in \mathcal{C}$ and  $\bu_i^\top \bu_i=1,$ for $i=1,\dots,n.$

\subsection{Optimization}\label{optimization}
In this section, we will discuss how to solve problem \eqref{e2}. We will employ a simple block-coordinate descent algorithm, which has guaranteed convergence and local optimality properties \citep{tseng2001convergence}. To obtain the coordinate descent steps, we need to solve the following three optimization problems.
\begin{enumerate}
    \item[P1] Fix, $\bw,\bmu$ and minimize $\sum_{i=1}^n \|\bw \circ (\bx_i-\bmu) - \bu_i\|_2^2, \hspace{5pt} \text{subject to  } \bu_i^\top \bu_i=1, \text{ } \forall i=1,\dots,n$, with respect to (w.r.t.) $\bU$.
    \item[P2] Fix, $\bw,\bU$ and minimize $\sum_{i=1}^n \|\bw \circ (\bx_i-\bmu) - \bu_i\|_2^2$, w.r.t. $\bmu$.
    \item[P3] Fix, $\bmu,\bU$ and minimize $\sum_{i=1}^n \|\bw \circ (\bx_i-\bmu) - \bu_i\|_2^2$, w.r.t. $\bw \in \mathcal{C}$.
\end{enumerate}
The problems P1-P3 admit closed form solutions. To solve P1, we need to solve the following optimization problem for each $i$:
\[\min_{\bu_i} \|\bw \circ (\bx_i-\bmu) - \bu_i\|_2^2, \text{ subject to } \|\bu_i\|^2_2=1. \]
The Lagrangian for this problem is given by 
\[\mathcal{L} = \|\bw \circ (\bx_i-\bmu) - \bu_i\|_2^2-\alpha(\|\bu_i\|^2_2-1)\]
Equation $\nabla_{\bu_i} \mathcal{L} = 0$ and enforcing the constraint, $\|\bu_i\|^2_2=1,$ it is easy to see that the solution of problem P1 is given by $\bu_i^\ast=\frac{\bw \circ(\bx_i-\bmu)}{\|\bw \circ(\bx_i-\bmu)\|_2}$. Problem P2 is solved by taking \[\mu_l^\ast=\frac{1}{n}\sum_{i=1}^n x_{il} -\frac{1}{n w_l} \sum_{i=1}^n u_{il},\] for $l=1,\dots,p$. Lastly, problem P3 can be solved by taking
$w_l^\ast= T\bigg(\frac{\sum_{i=1}^n u_{il}(x_{il}-\mu_l)}{\sum_{i=1}^n (x_{il}-\mu_l)^2},\lambda,\Lambda\bigg)$, for all $l=1,\dots,p$. Here,
\[T(x,\lambda,\Lambda)=\begin{cases}
\lambda & \text{if } x< \lambda\\
 \Lambda & \text{if } x>\Lambda\\
 x & \text{otherwise.}
\end{cases}
\]
 Algorithm \ref{algo} gives a formal description of the proposed coordinate descent algorithm, which we will call the Principal Elliptical Analysis (PEA) algorithm. One should note that this block-coordinate descent algorithm has guaranteed convergence properties \citep{tseng2001convergence}. 
\begin{algorithm*}
\caption{Principal Elliptical Analysis (PEA) Algorithm}\label{algo}
\begin{algorithmic}
\State \textbf{Input}:  $\bX \in \Real^{n \times p}$, $\lambda$, $\Lambda$, $\bw^{(0)}$, $\bmu^{(0)}$.
\State \textbf{Output}:  $\bw,\bmu$ and $\bU$. 
\Repeat
 \State \textbf{Step 1}: Update $\bU$ by $\bu_i^{(t+1)}=\frac{\bw^{(t)} \circ(\bx_i-\bmu^{(t)})}{\|\bw^{(t)} \circ(\bx_i-\bmu^{(t)})\|_2}$ for all $i=1,\dots,n$.\\
 \State \textbf{Step 2}: Update $\bmu$ by \[\mu_l^{(t+1)}=\frac{1}{n}\sum_{i=1}^n x_{il} -\frac{1}{n w_l^{(t)}} \sum_{i=1}^n u_{il}^{(t+1)},\] for all $l=1,\dots,p$.
 \State \textbf{Setp 3}: Update $\bw$ by taking \[w_l^{(t+1)}=T\bigg(\frac{\sum_{i=1}^n u_{il}^{(t+1)}(x_{il}-\mu_l^{(t+1)})}{\sum_{i=1}^n (x_{il}-\mu_l^{(t+1)})^2},\lambda,\Lambda\bigg).\]\\
 \Until{convergence}
\end{algorithmic}
\end{algorithm*}
\section{Theoretical Analysis}\label{theo}
\subsection{A Short Exposition on VC Theory}
In this section, we provide a short introduction to VC theory. Interested readers are referred to the book by \cite{devroye2013probabilistic} for a detailed exposition.

Let $(\Omega,\mathcal{B},P)$ be a probability space and let $\A \subset \mathcal{B}$ be a collection of $P$-measurable sets. In our analysis, we will assume that $\Omega=\Real^p$ and $\mathcal{B}$ is taken as  $\mathscr{B}(\Real^p)$, the Borel $\sigma$-algebra on $\Real^p$. Let $\mathcal{F}$ be a subset of $\Omega$ of size $n$. We say that a set $\A$ picks out $F \subseteq \mathcal{F}$ if there exists $A \in \A$ such that $A \cap \mathcal{F}=F$. Let $S(\A,\mathcal{F})$ be the number of subsets of $\mathcal{F}$ picked out by $\A$, i.e. \(S(\A,\mathcal{F})=|\{A \cap \mathcal{F}: A \in \A\}|\). We say that the set $\mathcal{F}$ is shattered by $\A$ if $S(\A,\mathcal{F})=2^n$, where $n=|\mathcal{F}|$. The shatter coefficient is defined as $S_n(\A)=\sup_{\mathcal{F}: |\mathcal{F}|=n} S(\A,\mathcal{F})$. One should note that $S_n(\A) \le 2^n$. The Vapnik-Chervonenkis (VC) dimension of $\A$, denoted by $V_\A$, is defined as the largest integer $n$, for which $S_n(\A)=2^n$, i.e. $V_{\A}=\sup \{n \in \mathbb{N}: S_n(\A)=2^n\}$.

Let $\bX_1,\dots,\bX_n$ be independent and identically distributed draws from $\sP$ and let $\Pn =\frac{1}{n}\sum_{i=1}^n \delta_{\bX_i}$ be the empirical measure based on the data $\bX_1,\dots,\bX_n$. Here $\delta_{\by}$ denotes the degenerate distribution at $\by$. We note that for any $A \subseteq \Real^p$, $\sP_n(A)=\frac{1}{n} \sum_{i=1}^n \one(\bX_i \in A)$, where $\one(\cdot)$ denotes the indicator function.  
We state the Vapnik Chervonenkis (VC) theorem as follows.

\textbf{Theorem} (\cite{vapnik1971uniform}): Let $\epsilon>0$, $n \epsilon^2 \ge 2$. Then
\[P\bigg(\sup_{A \in \A}|\Pn(A)-\sP(A)|>\epsilon\bigg)\le 8 S_n(A) \exp\bigg\{-\frac{n \epsilon^2}{32}\bigg\}.\]

\subsection{Large Sample Properties of PEA}
In this section, we study the large sample properties of the (global) optimal solutions produced by the Principal Elliptical Analysis (PEA) method. In particular, we will derive a uniform bound by appealing to the powerful Vapnik-Chervonenkis (VC) theory and show consistency of the obtained estimates.

We begin by assuming that $\bX_1,\dots,\bX_n$ are independent and identically distributed according to the distribution $\sP$. We note that for practical purposes, the support of $\sP$ can be assumed to be bounded, i.e. $\exists M>0$ such that $Support(\sP) \subset B(M)$. Here $B(M)=\{\bx \in \Real^p: \|\bx\|_2 \le M\}$ denotes the closed ball of radius $M$, centered at the origin. To simplify the objective function and facilitate the theoretical analysis, we eliminate $\bu_i$ from the objective. This is accomplished through representing $\bu_{i}=\frac{\bw \circ (\bX_i-\bmu)}{\|\bw \circ (\bX_i-\bmu)\|_2}$. Substituting the value of $\bu_i$ in equation \eqref{eqq}, we get \begin{equation}\label{eqq1}
f(\bw,\bmu)=\sum_{i=1}^n (\|\bw \circ (\bX_i-\bmu)\|_2-1)^2    
\end{equation} For notational simplicity, let $\Phi(\bmu,\bw,Q)=\int(\|\bw \circ (\bx-\bmu)\|_2-1)^2 dQ$ for any probability measure $Q$ (whenever $\Phi(\cdot,\cdot,Q)<\infty$). We note that the objective function \eqref{eqq1} can be written as $\Phi(\bmu,\bw,\Pn)$ (after dividing both sides of \eqref{eqq1} by $n$), where $\Pn =\frac{1}{n} \sum_{i=1}^n \delta_{\bX_i}$ is the empirical distribution function. Let $(\bmu_n,\bw_n)$ be the sample minimizers of $\Phi(\cdot,\cdot,\Pn)$, i.e. $\Phi(\bmu_n,\bw_n,\Pn) = \argmin\limits_{\bmu,\bw} \Phi(\bmu,\bw,\Pn)$. Let $(\bmu_0,\bw_0)$ be the population minimizers of $\Phi(\cdot,\cdot,\sP)$, i.e. $\Phi(\bmu_0,\bw_0,\sP) = \argmin\limits_{\bmu,\bw} \Phi(\bmu,\bw,\sP)$. Our goal is to show that $(\bmu_n,\bw_n) \to (\bmu_0,\bw_0)$, almost surely. Towards that, we make the following technical assumptions:
\begin{itemize}
    \item[A1]  The support of $\sP$ is bounded, i.e. $\exists M>0$ such that $Support(\sP) \subset B(M)$.
    \item[A2] For all $\epsilon>0$, there exists $\eta>0$ such that if $\|(\bw,\bmu)^\top - (\bw_0,\bmu_0)^\top\|_2 > \epsilon$ then $\Phi(\bmu,\bw,\sP)>\Phi(\bmu_0,\bw_0,\sP)+\eta$. 
\end{itemize}
The first step towards the consistency result is to bound $(\bmu_n,\bw_n)$ in some compact set $\forall n \in \mathbb{N}$. Since $\bw$ is already bounded in $\mathcal{C}$, it is enough to bound $\bmu_n$ in some compact ball. This is achieved in Theorem \ref{thm1}. Due to space restrictions, the proof is provided in the supplementary document.
 \begin{thm}\label{thm1}
 There exists $K>0$ such that $\bmu_n \in B(M+K)$.
 \end{thm}
In Theorem \ref{vc1}, we derive a uniform concentration bound for  $\Phi(\cdot,\cdot,\Pn)$. This theorem plays a crucial role in the proof of Theorem \ref{vc2}.
\begin{thm}\label{vc1}
Suppose assumption A1 holds. Also assume that $n > 4p+2$ and $n \epsilon^2 \ge 2$. Then,
\[
     P\bigg(\sup\limits_{\bmu \in B(M+K), \bw \in \mathcal{C}} |\Phi(\bmu,\bw,\Pn) - \Phi(\bmu,\bw,\sP)|> \epsilon\bigg)
     \le  16 \bigg(\frac{n e }{2p+1}\bigg)^{2p+1} \exp\bigg\{-\frac{n \epsilon^2}{32 D^2}\bigg\},
\]
where $D=\sup_{\bx \in B(M), \bmu \in B(M+K), \bw \in \mathcal{C}}(\|\bw \circ (\bx-\bmu)\|_2-1)^2$.
\end{thm}
\begin{proof}
Let $\phi_{\bmu,\bw}(\bx)=(\|\bw \circ (\bx-\bmu)\|_2-1)^2$. Clearly the family of functions $\{\phi_{\bmu,\bw}(\cdot): \bmu \in B(M+K), \bw \in \mathcal{C}\}$ is equi-continuous on the compact set $B(M)$ and hence bounded above by $D$ (say). We now observe that

\begingroup
\allowdisplaybreaks
\begin{align*}
    & \sup\limits_{\bmu \in B(M+K), \bw \in \mathcal{C}} |\Phi(\bmu,\bw,\Pn) - \Phi(\bmu,\bw,\sP)|\\
     = & \sup\limits_{\bmu \in B(M+K), \bw \in \mathcal{C}} \bigg| \frac{1}{n} \sum_{i=1}^n\phi_{\bmu,\bw}(\bX_i) - E_{\bX \sim \sP }[\phi_{\bmu,\bw}(\bX)]\bigg|\\
     = & \hspace{-0.3cm} \sup\limits_{\bmu \in B(M+K),\bw \in \mathcal{C}} \bigg| \int_{0}^D\frac{1}{n}\sum_{i=1}^n \one(\phi_{\bmu,\bw}(\bX_i)>t) dt - \int_{0}^D P(\phi_{\bmu,\bw}(\bX)>t) dt\bigg|\\
     \le &  D \sup\limits_{\bmu \in B(M+K), \bw \in \mathcal{C}, t \ge 0 } \bigg| \frac{1}{n} \sum_{i=1}^n \one(\phi_{\bmu,\bw}(\bX_i)>t) - P(\phi_{\bmu,\bw}(\bX)>t) \bigg|\\
    \le & D \sup_{A \in \mathcal{A}} |\Pn(A)-\sP(A)|
\end{align*}
\endgroup
Here $\mathcal{A}=\bigg\{\{\bx \in B(M):\phi_{\bmu,\bw}(\bx)>t \}: \bmu \in B(M+K), \bw \in \mathcal{C} \text{ and } t>0\bigg\}$. It is easy to observe that $\mathcal{A} \subseteq \mathcal{A}_1 \cup \mathcal{A}_2$, where, $\mathcal{A}_1$ and $\mathcal{A}_2$ are defined as follows,
\[
\mathcal{A}_1=\bigg\{\{\bx \in B(M):\|\bw \circ (\bx-\bmu)\|_2-1 > \sqrt{t} \}: \bmu \in B(M+ K),
\bw \in \mathcal{C} \text{ and } t>0\bigg\}    
\]
\[
\mathcal{A}_2=\bigg\{\{\bx \in B(M):\|\bw \circ (\bx-\bmu)\|_2-1 < -\sqrt{t} \}: \bmu \in B(M +K), \bw \in \mathcal{C} \text{ and } t>0\bigg\}    
\]

To find the VC-dimension of $\mathcal{A}_1$, we observe that if $\|\bw \circ (\bx-\bmu)\|_2-1 > \sqrt{t}$, then $\|\bw \circ (\bx-\bmu)\|_2^2 > (\sqrt{t}+1)^2$. Let $\bz=(x_1,\dots,x_p,x_1^2,\dots,x_p^2)$. On $\Real^{2p}$, we define the set,
\[\mathcal{G}:=\bigg\{\{\bz \in \Real^{2p}: \bbeta^\top\bz > \gamma\}: \beta \in \Real^{2p}, \gamma \in \Real\bigg\}.\]
Observe that any member of $\mathcal{A}_1$ corresponds to a member of $\mathcal{G}$, for which, $\beta=(w_1^2,\dots,w_p^2, -2 \mu_1 w_1^2,\dots,-2 \mu_p w_p^2)$ and $\gamma=(\sqrt{t}+1)^2-\sum_{l=1}^p w_l^2 \mu_l^2$. This implies that $ V_{\mathcal{A}_1} \le V_{\mathcal{G}}= 2p+1$. Thus, in light of Theorem 13.3 of \cite{devroye2013probabilistic}, $S_n(\mathcal{A}_1) \le \bigg(\frac{n e }{2p+1}\bigg)^{2p+1}$. By a similar argument, one can show that $S_n(\mathcal{A}_2) \le \bigg(\frac{n e }{2p+1}\bigg)^{2p+1}$. Thus appealing to Theorem 13.5 $(i)$ of \cite{devroye2013probabilistic}, we get, 
\(S_n(\mathcal{A}) \le 2 \bigg(\frac{n e }{2p+1}\bigg)^{2p+1}\).
From the VC-theorem, we observe that,
\begin{align*}
     P\left(\sup\limits_{\bmu \in B(M+K), \bw \in \mathcal{C}} |\Phi(\bmu,\bw,\Pn) - \Phi(\bmu,\bw,\sP)|> \epsilon\right)
    & \le P\left(  \sup_{A \in \mathcal{A}} |\Pn(A)-\sP(A)|> \frac{D}{\epsilon}\right)\\ & \le 8 S_n(\mathcal{A}) \exp\left\{-\frac{n \epsilon^2}{32 D^2}\right\}\\
    &\le 16 \left(\frac{n e }{2p+1}\right)^{2p+1} \exp\left\{-\frac{n \epsilon^2}{32 D^2}\right\}.
\end{align*}
\end{proof}
We are now ready to bound the probability of $|\Phi(\bmu_n,\bw_n,\sP)-\Phi(\bmu_0,\bw_0,\sP)|$ being greater than any infinitesimal quantity $\epsilon>0$. 
\begin{thm}\label{vc2}
Suppose assumption A1 holds. Also assume that $n > 4p+2$ and $n \epsilon^2 \ge 2$. Then,
\[
P(|\Phi(\bmu_n,\bw_n,\sP)-\Phi(\bmu_0,\bw_0,\sP)|> \epsilon) \le 16 \bigg(\frac{n e }{2p+1}\bigg)^{2p+1} \exp\bigg\{-\frac{n \epsilon^2}{128 D^2}\bigg\}.    
\]
\end{thm}
\begin{proof}
We consider the following difference 
\begin{align*}
   & \Phi(\bmu_n,\bw_n,\sP)-\Phi(\bmu_0,\bw_0,\sP)\\
   & = \Phi(\bmu_n,\bw_n,\sP) - \Phi(\bmu_n,\bw_n,\Pn) + \Phi(\bmu_n,\bw_n,\Pn) - \Phi(\bmu_0,\bw_0,\sP)\\
   & \le \Phi(\bmu_n,\bw_n,\sP) - \Phi(\bmu_n,\bw_n,\Pn) + \Phi(\bmu_0,\bw_0,\Pn)- \Phi(\bmu_0,\bw_0,\sP)\\
   & \le 2 \sup\limits_{\bmu \in B(M+K), \bw \in \mathcal{C}} |\Phi(\bmu,\bw,\Pn) - \Phi(\bmu,\bw,P)|
\end{align*}
Thus,
\begin{align*}
 P(|\Phi(\bmu_n,\bw_n,\sP)-\Phi(\bmu_0,\bw_0,\sP)|> \epsilon) 
& \le P\bigg(2 \sup\limits_{\bmu \in B(M+K), \bw \in \mathcal{C}} |\Phi(\bmu,\bw,\Pn) - \Phi(\bmu,\bw,\sP)|> \epsilon \bigg)    \\
& \le 16 \bigg(\frac{n e }{2p+1}\bigg)^{2p+1} \exp\bigg\{-\frac{n \epsilon^2}{128 D^2}\bigg\}.
\end{align*}
The last inequality follows from Theorem \ref{vc1}.
\end{proof}
\begin{thm}
Under assumptions A1 and A2, $\Phi(\bmu_n,\bw_n,\sP) \to \Phi(\bmu_0,\bw_0,\sP)$, $a.s. [P]$. Moreover, $\|(\bw_n,\bmu_n)^\top - (\bw_0,\bmu_0)^\top\|_2 \to 0$ almost surely as $n \to \infty$.
\end{thm}
\begin{proof}
The almost sure convergence can be observed from the fact that 
\[ \sum_{n=1}^\infty P\big(\Phi(\bmu_n,\bw_n,\sP)-\Phi(\bmu_0,\bw_0,\sP) > \epsilon\big) < \infty,\]
for all $\epsilon>0$. Thus, $\Phi(\bmu_n,\bw_n,\sP) \to \Phi(\bmu_0,\bw_0,\sP)$, $a.s. [P]$.
This implies that for all $\eta>0$, $\Phi(\bmu_n,\bw_n,\sP) \le \Phi(\bmu_0,\bw_0,\sP) + \eta$, eventually. From assumption A2, the result, thus, follows.
\end{proof}

\section{Extension to Clustering}
In this section, we discuss applications of PEA in clustering datasets. Suppose we have $n$ many data points, $\mathcal{X}=\{\bx_1,\dots,\bx_n\} \subset \Real^p$, to be partitioned into $k$ many disjoint clusters. The clustering problem is formulated as the minimization of the following objective function:
\begin{equation}\label{cl_obj}
P(\bW,\bM,\bc )=\sum_{j=1}^k \sum_{i: c_i=j} (\|\bw_j \circ (\bx_i-\bmu_j)\|_2-1)^2 , 
\end{equation}
where $\bW_{k \times p}=(\bw_1:\bw_2:\dots:\bw_k)^\top$, $\bM_{k \times p}=(\bmu_1:\bmu_2:\dots:\bmu_k)^\top$ and $\bc=(c_1,\dots, c_n)$ denote the cluster indicator vector, $c_i=j$ indicates that the $i$-th datapoint belongs to the $j$-th cluster. Here $\bmu_j$ and $\bw_j$ denotes the center and inverse half axis lengths for the $j$-th cluster. We employ a simple block-coordinate descent algorithm, similar to Lloyd's method \citep{lloyd1982least}, to minimize the objective function \eqref{cl_obj}. Algorithm \ref{algo_cluster} gives a formal description of the algorithm.
\begin{algorithm*}[ht]
\caption{Principal Elliptical Analysis (PEA) Clustering Algorithm}\label{algo_cluster}
\begin{algorithmic}
\State \textbf{Input}:  $\bX \in \Real^{n \times p}$.
\State \textbf{Output}:  $\bc$ and $\bw_j,\bmu_j$, $j=1,\dots,k$. 
\State Initialize $\bc$ by $k$-means.
\Repeat
 \State \textbf{Step 1}: Update $\bU$ by $\bu_i^{(t+1)}=\frac{\bw^{(t)} \circ(\bx_i-\bmu^{(t)})}{\|\bw^{(t)} \circ(\bx_i-\bmu^{(t)})\|_2}$ for all $i=1,\dots,n$.\\
 \State \textbf{Step 2}: Update $\bmu$ by $\mu_l^{(t+1)}=\frac{1}{n}\sum_{i=1}^n x_{il} -\frac{1}{n w_l^{(t)}} \sum_{i=1}^n u_{il}^{(t+1)}$, for all $l=1,\dots,p$.
 \State \textbf{Setp 3}: Update $\bw$ by taking $w_l^{(t+1)}=\frac{\sum_{i=1}^n u_{il}^{(t+1)}(x_{il}-\mu_l^{(t+1)})}{\sum_{i=1}^n (x_{il}-\mu_l^{(t+1)})^2}$.\\
 \State \textbf{Step 4}: Update $\bc$ by $c_i^{(t+1)}=j$, where $j=\arg\min_{1 \le m \le k} (\|\bw_m^{(t+1)}\circ (\bx_i-\bmu_m^{(t+1)})\|_2-1)^2$.
 \Until{convergence}
\end{algorithmic}
\end{algorithm*}
\subsection{Theoretical Analysis}
One can analyze the large sample solutions of the obtained estimates via VC theory, similar to that derived in Section \ref{theo}. We only state the relevant theorems. The detailed proofs can be found in the Appendix. 

Similar to the assumptions made in Section \ref{theo}, we assume that $\bX_1,\dots,\bX_n$ are independently drawn from distribution $P$. For our theoretical analysis, note that we can simplify the objective function \eqref{cl_obj} as 
\begin{equation}\label{eqq2}
P(\bW,\bM)=\sum_{i=1}^n \min_{1\leq j\leq k} (\|\bw_j \circ (\bX_i-\bmu_j)\|_2-1)^2. 
\end{equation}
Again to simplify the notations, let $\Psi(\bW,\bM,Q)=\int \min_{1 \le j \le k} (\|\bw_j \circ (\bx-\bmu_j)\|_2-1)^2 dQ$. 
Let $(\bM_n^\ast,\bW_n^\ast)$ be the sample minimizers of $\Psi(\cdot,\cdot,\Pn)$ and  $(\bar{\bM},\bar{\bW})$ be the population minimizers of $\Psi(\cdot,\cdot,\sP)$. Our goal is to show that $(\bM_n^\ast,\bW_n^\ast) \to (\bar{\bM},\bar{\bW})$, almost surely. Under standard regularity conditions, we can prove the following theorem.
\begin{thm}
Suppose assumption A1 holds. Also assume that $n > 4p+2$ and $n \epsilon^2 \ge 2$. Then,
\begin{align*}
    &P(|\Psi(\bM_n^\ast,\bW_n^\ast,\sP)-\Psi(\bar{\bM},\bar{\bW},\sP)|> \epsilon) \le 8 \times 2^k \bigg(\frac{n e }{2p+1}\bigg)^{k(2p+1)} \exp\bigg\{-\frac{n \epsilon^2}{128 D^2}\bigg\},
\end{align*}
where $D$ is as defined in Theorem \ref{vc1}.
\end{thm}
\subsection{Experimental Results}\label{expo}
In this section, we assess the performance of the PEA clustering algorithm compared to some baseline and state-of-the-art clustering techniques. The microarray datasets (the first 5) are collected from \cite{jin2016influential}\footnote{\url{http://www.stat.cmu.edu/~jiashun/Research/software/GenomicsData/}}. The other real datasets are collected from \cite{alcala2011keel}\footnote{\url{https://sci2s.ugr.es/keel/datasets.php}}.  For a comparative study, we compare our method with classic $k$-means \citep{macqueen1967some}, Weighted $k$-means ($W$-$k$-means) \citep{huang2005automated}, Sparse $k$-means \citep{witten2010framework} and Influential Features PCA with Higher Criticism thresholding \citep{jin2016influential}.  As cluster validation indices, we compute the Normalized Mutual Information (NMI) \citep{vinh2010information}, Adjusted Rand Index \citep{hubert1985comparing} and the Clustering Error Rate (CER) \citep{kondo2016rskc} between the ground truth and the partitioning obtained by each of the algorithms. An NMI or ARI value of $1$ indicates perfect clustering, while a value of $0$ indicates a complete mismatch between the obtained partition and the ground truth. On the other hand, the CER measures $0$ for a perfect partitioning while it measures $1$ for complete mismatch.
\begin{table*}[!h]
 \caption{Clustering Performance of Different Algorithms on Real Data Benchmarks.} {Here, $n=$ no. of samples, $p=$ no. of features, $k=$ no. of clusters.}
 \label{cluster}
    \centering
    \begin{tabular}{c c c c c c c c}
    \hline
       Datasets  &  Index & $k$-Means & \makecell{IF-HCT\\$k$-means} & $W$-$k$-means & \makecell{Sparse\\$k$-Means} & IF-HCT-PCA & PEA\\
       \hline
       Brain & NMI & 0.5443 & 0.5764 & 0.3903 & 0.6259 & 0.5574 & \textbf{0.6431}\\
        $n=42 , p=5597$ & ARI & 0.4363 & 0.4739 & 0.1591 & 0.5471 & 0.4522 & \textbf{0.5440}\\
        $k=5$ & CER & 0.286 & 0.252 & 0.324 & 0.152 & 0.262 & \textbf{0.136}\\
       \hline
       Leukemia & NMI & 0.6421 & 0.7126 & 0.1536 & 0.6813 & 0.7321 & \textbf{0.8056}\\
        $n=72 , p=3571$ & ARI & 0.6826 & 0.7456 & 0.2127 & 0.1398 & 7929 & \textbf{0.8897}\\
        $k=2$ & CER & 0.278 & 0.223 & 0.430 & 0.158 & 0.069 & \textbf{0.054}\\
       \hline
       Lymphoma & NMI & 0.5801 & 0.6581 & 0.5786 & 0.7019 & 0.7585 & \textbf{0.9254}\\
        $n=62 , p=4026$  & ARI &  0.3944 & 0.4817 & 0.3881 & 0.8306 & 0.8422 & \textbf{0.9471}\\
        $k=3$ & CER & 0.387 & 0.259 & 0.307 & 0.184 & 0.069 & \textbf{0.026}\\
       \hline
       SRBCT & NMI & 0.2108 & 0.3686 & 0.0857 &  0.4578 & 0.1032 & \textbf{0.6313}\\
        $n=63, p=2308$ & ARI & 0.0733 & 0.1720 & 0.0160 & 0.4079 & 0.0897 & \textbf{0.4468}\\
        $k=4$ & CER & 0.556 & 0.330 & 0.461 & 0.297 & 0.444 & \textbf{0.221}\\
       \hline
       Colon & NMI & 0.0053 & 0.0124 & 0.0045 & 0.1343 & 0.1873 & \textbf{0.4295}\\
       $n=62 , p=2000$& ARI & 0.0160 & 0.0182 & 0.0009 & 0.0880 & 0.1133 & \textbf{0.5420}\\
       $k=2$ & CER & 0.443 & 0.4892 & 0.506 & 0.432 & 0.403 & \textbf{0.228}\\
       \hline
       Wine & NMI & 0.4283 & 0.4287 & 0.1257 & 0.7838 & 0.1323 & \textbf{0.8186}\\
       $n=178 , p=13$ & ARI & 0.3642 & 0.3711 & 0.1071 & 0.8064 & 0.1143 & \textbf{0.8300}\\
        $k=3$ & CER & 0.280 & 0.281 & 0.403 & 0.329 & 0.401 & \textbf{0.052}\\
       \hline
       WDBC & NMI & 0.4636 & 0.4648 & 0.0003 & 0.4648 & 0.4674 & \textbf{0.6447}\\
        $n=569 , p=30$ & ARI & 0.4904 & 0.4914 & 0.0046 & 0.2496 & 0.4914 &\textbf{ 0.7113}\\
        $k=2$ & CER & 0.248 & 0.249 & 0.493 & 0.491 & 0.250 & \textbf{0.148}\\
       \hline
       Movement Libras & NMI & 0.5398 & 0.5655 & 0.4413 & 0.0666 & 0.1482 & \textbf{0.5956}\\
        $n=360, p=90$ & ARI & 0.2492 & 0.2659 & 0.1732 & 0.0205 & 0.2133 & \textbf{0.3132}\\
        $k=15$ & CER & 0.099 & 0.107 & 0.110 & 0.493 & 0.245 & \textbf{0.090}\\
       \hline
    \end{tabular}
\end{table*}

As high-dimensional datasets often contain a large number of uninformative features \citep{jin2016influential}, we perform an influential feature selection \citep{jin2016influential} before running the PEA algorithm. We also apply this same preprocessing step before running k-means, leading to the k-means IF-HCT algorithm.

We run each algorithm $20$ times and report their average clustering index values in Table \ref{cluster}. The best performing algorithm in each row is boldfaced. The results in Table \ref{cluster} indicate that clustering under PEA resembles the ground truth more closely compared to its competitors.

\section{Conclusion}
\label{con}
This paper introduces a simple and computationally efficient but flexible extension of PCA and k-means, relying on the idea of locally approximating data within clusters as falling near the surface of an unknown ellipsoid.   The resulting Principal Elliptical Analysis (PEA) approach has surprisingly good performance in a broad variety of clustering algorithms we have considered.  We have additionally provided strong theory support for PEA estimators through VC theory.  PEA can also be used to produce elliptical principal component scores to obtain a rich class of geometric extensions of PCA for interpretation. Extensions to broader classes of quadratic surfaces are an interesting future direction.

\section*{Acknowledgement}
The following grants are acknowledged for the financial support provided for this research: Office of Naval Research W911NF-16-1-0544 and National Institute of Health R01ES027498.

\bibliographystyle{apalike}

\appendix
\section{Proof of Theorem 1}
\paragraph{Theorem 1.}
There exists $K>0$ such that $\bmu_n \in B(M+K)$.
\begin{proof}
We choose $K=\frac{\sqrt{(M-1)^2+1}}{\lambda} +1$. Without the loss of generality, we will also assume that $\Lambda \ge M \ge 1$. Assume the contrary. We assume that $\|\bmu_n\|_2 > M+K$. We begin by observing that $\|\bw \circ (\bx-\bmu_n)\|_2 \ge \lambda \|\bx-\bmu_n\|_2 \ge \lambda (\|\bmu_n\|_2-\|\bx\|_2) \ge \lambda K \ge 1 $. We note that,
\begin{align}
    \Phi(\bmu_n,\bw_n, \sP_n) & = \int(\|\bw_n \circ (\bx-\bmu_n)\|_2-1)^2 d\sP_n \nonumber\\
    & \ge \int(K \lambda-1)^2 d\sP_n  \nonumber \\
    & = (K \lambda-1)^2. \label{1}
\end{align}
We observe that \(0 \le \|\bx\|_2 \le M \) implies \[(\|\bx\|_2-1)^2 \le \max\{(M-1)^2,1\} \le (M-1)^2+1.\] Thus, we observe that,
\begin{align}
    \Phi(\mathbf{0}_p,\mathbf{1}_p, \sP_n) & = \int(\|\mathbf{1}_p \circ \bx\|_2-1)^2 d\sP_n \nonumber\\
    & = \int(\|\bx\|_2-1)^2 d\sP_n \nonumber\\
    & \le \int[(M-1)^2+1] d\sP_n \nonumber \\
    & = (M-1)^2+1. \label{2}
\end{align}
Thus RHS of equation \eqref{1} is always greater than the RHS of equation \eqref{2}, due to the choice of $K$. Thus, $\Phi(\bmu_n,\bw_n, P_n) < \Phi(\mathbf{0}_p,\mathbf{1}_p, P_n)$, which gives us a contradiction. 
\end{proof}
\section{Theoretical Properties of PEA Clustering}
In this section, we discuss the large sample properties of the (global) minimimal solutions of our clustering objective (\textcolor{blue}{5}) (of the main paper). Similar to the assumptions made in Section \textcolor{blue}{4}, we assume that $\bX_1,\dots,\bX_n$ are independently drawn from distribution $\sP$. For our theoretical analysis, note that we can simplify the objective function (\textcolor{blue}{5}) as 
\begin{equation}\label{eqq3}
P(\bW,\bM)=\sum_{i=1}^n \min_{1\leq j\leq k} (\|\bw_j \circ (\bx_i-\bmu_j)\|_2-1)^2. 
\end{equation}
Again to simplify the notations, let $\Psi(\bM,\bW,Q)=\int \min_{1 \le j \le k} (\|\bw_j \circ (\bx-\bmu_j)\|_2-1)^2 dQ$, for any $\bW \in [\lambda,\Lambda]^{k \times p}$, $\bM \in \Real^{k \times p}$ and any probability measure $Q$, for which $\Psi(\cdot,\cdot,Q)$ is defined. Let $(\bM_n^\ast,\bW_n^\ast)$ be the sample minimizers of $\Psi(\cdot,\cdot,\sP_n)$, i.e. $\Phi(\bM_n^\ast,\bW_n^\ast,\sP_n) = \arg\min\limits_{\bM,\bW} \Psi(\bM,\bW,\sP_n)$. Let $(\bar{\bM},\bar{\bW})$ be the population minimizers of $\Psi(\cdot,\cdot,\sP)$, i.e. $\Psi(\bar{\bM},\bar{\bW},\sP) = \arg\min\limits_{\bM,\bW} \Psi(\bM,\bW,\sP)$. Our goal is to show that $(\bM_n,\bW_n) \to (\bar{\bM},\bar{\bW})$, almost surely. Since renaming the clusters does not affect the value of the objective function, we define the distance between two sets of parameters $(\bM_1,\bW_1)$ and $(\bM_2,\bW_2)$ as follows:
\begin{align*}
& dist \bigg( (\bM_1,\bW_1) , (\bM_2,\bW_2) \bigg)\\
 = & \inf_{O \in \mathcal{O}_k} \|\bW_1-O\bW_2\|^2_F + \|\bM_1-O\bM_2\|^2_F.    
\end{align*}
Here $\mathcal{O}_k$ is the set of all $k \times k$ real permutation matrices. We say that $(\bM_n^\ast,\bW_n^\ast) \to (\bar{\bM},\bar{\bW})$ if $dist\big((\bM_n^\ast,\bW_n^\ast), (\bar{\bM},\bar{\bW})\big) \to 0$.
Towards proving the almost sure convergence, we make the following technical assumptions:
\begin{itemize}
    \item[B1]  The support of $P$ is bounded, i.e. $\exists M>0$ such that $Support(P) \subset B(M)$.
    \item[B2] For all $\epsilon>0$, there exists $\eta>0$ such that if $dist\big((\bM,\bW), (\bar{\bM},\bar{\bW})\big) > \epsilon$ implies that $\Psi(\bM,\bW,\sP)>\Psi(\bar{\bM},\bar{\bW},\sP)+\eta$. 
\end{itemize}
Without he loss of generality, we will assume that $\Lambda \ge M \ge 1$. As already discussed in Section \textcolor{blue}{4} we first need to bound $(\bM_n^\ast,\bW_n^\ast)$ in some compact set $\forall n \in \mathbb{N}$. Since $\bW$ is already bounded in $[\lambda,\Lambda]^{k \times p}$, it is enough to bound $\bM_n^\ast$ in a compact set. This is achieved in Theorem \ref{c1}. The Proof is very similar to that of Theorem \ref{thm1}.
\begin{thm}\label{c1}
There exists $K>0$, such that, $\bM_n^\ast \in B(M+K)^k$, for all $n \in \mathbb{N}$. 
\end{thm}
\begin{proof}
We will take $K=\frac{\sqrt{(M-1)^2+1}}{\lambda} +1$ and assume the contrary. We denote the $j$-th row of $\bM_n^\ast$ and $\bW^\ast_n$ by $\bmu_j^\ast$ and $\bw_j^\ast$, respectively.  and Suppose there exists $j \in \{1,\dots,k\}$, such that, $\|\bmu_j^\ast\|_2 \ge K$. If $c_i=j$, $\|\bw_j^\ast \circ (\bx-\bmu_j^\ast)\|_2 \ge \lambda \|\bx_i-\bmu_j^\ast\|_2 \ge \lambda (\|\bmu_j^\ast\|_2-\|\bx\|_2) \ge \lambda K  \ge 1$. Thus, $(\|\bw_j^\ast \circ (\bx-\bmu_j^\ast)\|_2 -1 )^2 \ge (\lambda K -1 )^2$. As observed in the proof of Theorem \ref{thm1}, $(\|\mathbf{1}_p \circ (\bx-\mathbf{0}_p)\|_2 -1 )^2 \le (M-1)^2+1$. Our choice of $K$ implies that $(\|\bw_j^\ast \circ (\bx-\bmu_j^\ast)\|_2 -1 )^2 \ge (\|\mathbf{1}_p \circ (\bx-\mathbf{0}_p)\|_2 -1 )^2$. Thus, replacing $\bmu_j^\ast$ by $\mathbf{0}_p$ and $\bw_j^\ast$ by $\mathbf{1}_p$, decreases the value of the objective function (\textcolor{blue}{5}) (of the main paper), which gives us a contradiction. 
\end{proof}
\begin{thm}\label{vcap}
Suppose assumption B1 holds. Also assume that $n > 4p+2$ and $n \epsilon^2 \ge 2$. Then,
\begin{align*}
& P\bigg(\sup\limits_{\bM \in B(M+K)^k, \bW \in [\lambda, \Lambda]^{k \times p}} |\Psi(\bM,\bW,\sP_n) - \Psi(\bM,\bW,\sP)|> \epsilon\bigg) \le 8 \times 2^k \bigg(\frac{n e }{2p+1}\bigg)^{k(2p+1)} \exp\bigg\{-\frac{n \epsilon^2}{32 D^2}\bigg\},    
\end{align*}
where, $D=\sup_{\bx \in B(M), \bmu \in B(M+K), \bw \in \mathcal{C}}(\|\bw \circ (\bx-\bmu)\|_2-1)^2$.
\end{thm}
\begin{proof}
Let $\psi_{\bM,\bW}(\bx)=\min_{1 \le j \le k}(\|\bw_j \circ (\bx-\bmu_j)\|_2-1)^2$, which is clearly bounded by $D$. In the following analysis, $\bX$ is a random variable having distribution $\sP$. We now observe that,
\begin{align*}\small
    & \sup\limits_{\bM \in B(M+K)^k, \bW \in [\lambda, \Lambda]^{k \times p}} |\Psi(\bM,\bW,\sP_n) - \Psi(\bM,\bW,\sP)|\\
     = & \sup\limits_{\bM \in B(M+K)^k, \bW \in [\lambda, \Lambda]^{k \times p}} \bigg| \frac{1}{n} \sum_{i=1}^n\psi_{\bM,\bW}(\bX_i) - E[\psi_{\bM,\bW}(\bX)]\bigg|\\
    = & \sup\limits_{\bM, \bW } \bigg| \int_{0}^D\bigg(\frac{1}{n} \sum_{i=1}^n \one(\psi_{\bM,\bW}(\bX_i)>t) - P(\psi_{\bM,\bW}(\bX)>t)\bigg) dt\bigg|\\
     \le & D \sup\limits_{\bM, \bW, t \ge 0} \bigg| \frac{1}{n} \sum_{i=1}^n \one(\psi_{\bM,\bW}(\bX_i)>t) - P(\psi_{\bM,\bW}(\bX)>t) \bigg|\\
     \le & D \sup_{A \in \mathcal{A}} |\sP_n(A)-\sP(A)|.
\end{align*}
Here $\mathcal{A}=\bigg\{\{\bx \in B(M):\psi_{\bM,\bW}(\bx)>t \}: \bM \in B(M+K)^k, \bW \in [\lambda, \Lambda]^{k \times p} \text{ and } t>0\bigg\}$. We note the following:
\begin{align*}
    &\psi_{\bM,\bW}(\bX_i)>t\\
    \iff & \min_{1 \le j \le k}(\|\bw_j \circ (\bx-\bmu_j)\|_2-1)^2>t\\
    \iff & (\|\bw_j \circ (\bx-\bmu_j)\|_2-1)^2>t, \, \forall \, j=1,\dots,k.
\end{align*}
Thus, it is easy to observe that $\mathcal{A} \subseteq \{V_1\cap \dots \cap V_k: V_1, \dots, V_k \in \mathcal{V} \}$, where,
\begin{align*}
\mathcal{V}=&\bigg\{\{\bx \in B(M):(\|\bw \circ (\bx-\bmu)\|_2-1)^2 > t \}: \bmu \in B(M+K), \bw \in \mathcal{C} \text{ and } t>0\bigg\}.    
\end{align*}
We've already seen that $S_n(\mathcal{V}) \le 2 \bigg(\frac{n e }{2p+1}\bigg)^{2p+1}$. Thus appealing to  Theorem 13.5 $(iii)$ of \cite{devroye2013probabilistic}, we get, 
\[S_n(\mathcal{A}) \le \bigg[2 \bigg(\frac{n e }{2p+1}\bigg)^{2p+1}\bigg]^k. \]
From the VC-theorem, we observe that,
{\small
\begin{align*}
    & P\bigg(\sup\limits_{\bmu \in B(M+K), \bW \in [\lambda, \Lambda]^{k \times p}} |\Psi(\bM,\bW,\sP_n) - \Phi(\bM,\bW,\sP)|> \epsilon\bigg)\\
     \le & P\bigg(  \sup_{A \in \mathcal{A}} |P_n(A)-P(A)|> \frac{D}{\epsilon}\bigg)\\
     \le & 8 S_n(\mathcal{A}) \exp\bigg\{-\frac{n \epsilon^2}{32 D^2}\bigg\}\\
    \le & 8 \times 2^k \bigg(\frac{n e }{2p+1}\bigg)^{k(2p+1)} \exp\bigg\{-\frac{n \epsilon^2}{32 D^2}\bigg\}.
\end{align*}
}%
\end{proof}
\begin{thm}
Suppose assumption A1 holds. Also assume that $n > 4p+2$ and $n \epsilon^2 \ge 2$. Then,
\begin{align*}
&P(|\Psi(\bM_n^\ast,\bW_n^\ast,\sP)-\Psi(\bar{\bM},\bar{\bW},\sP)|> \epsilon)\\
\le & 8 \times 2^k \bigg(\frac{n e }{k(2p+1)}\bigg)^{2p+1} \exp\bigg\{-\frac{n \epsilon^2}{128 D^2}\bigg\}.    
\end{align*}
\end{thm}
\begin{proof}
We consider the following difference 
\begin{align*}
   & \Psi(\bM_n^\ast,\bW_n^\ast,\sP)-\Psi(\bar{\bM},\bar{\bW},\sP)\\
    = & \Psi(\bM_n^\ast,\bW_n^\ast,\sP) - \Psi(\bM_n^\ast,\bW_n^\ast,\sP_n) + \Psi(\bM_n^\ast,\bW_n^\ast,\sP_n) - \Psi(\bar{\bM},\bar{\bW},\sP)\\
    \le & \Psi(\bM_n^\ast,\bW_n^\ast,\sP) - \Psi(\bM_n^\ast,\bW_n^\ast,\sP_n) + \Psi(\bar{\bM},\bar{\bW},\sP_n) - \Psi(\bar{\bM},\bar{\bW},\sP)\\
    \le & 2 \sup\limits_{\bM \in B(M+K), \bW \in [\lambda, \Lambda]^{k \times p}} |\Psi(\bM,\bW,\sP_n) - \Psi(\bM,\bW,\sP)|.
\end{align*}
Thus,
{\small
\begin{align*}
& P(|\Psi(\bM_n^\ast,\bW_n^\ast,\sP)-\Psi(\bar{\bM},\bar{\bW},\sP)|> \epsilon) \\
 \le & P\bigg(2 \sup\limits_{\bM \in B(M+K), \bW \in \mathcal{C}} |\Psi(\bM,\bW,\sP_n) - \Psi(\bM,\bW,\sP)|> \epsilon \bigg)    \\
 \le & 8 \times 2^k \bigg(\frac{n e }{2p+1}\bigg)^{k(2p+1)} \exp\bigg\{-\frac{n \epsilon^2}{128 D^2}\bigg\}.
\end{align*}
}%
The last inequality follows from Theorem \ref{vcap}.
\end{proof}

\end{document}